\documentclass[11pt]{article}
\usepackage[utf8]{inputenc}
\usepackage{amsmath,amssymb}
\usepackage[twoside,letterpaper,margin=1in]{geometry}
\usepackage{hyperref,algorithm,algorithmic}
\usepackage{graphicx,url,amsmath,amsthm,amsfonts,amssymb,subfigure,bbm,bm}
\usepackage{color}
\usepackage{verbatim}
\usepackage{thmtools}
\usepackage{thm-restate}
\usepackage{wrapfig}

\usepackage{pdfsync}
\newcommand{\executeiffilenewer}[3]{%
\ifnum\pdfstrcmp{\pdffilemoddate{#1}}%
{\pdffilemoddate{#2}}>0%
{\immediate\write18{#3}}\fi%
}

\graphicspath{{images/}}

\newcounter{sideremark}

\newif\ifcomments
\commentstrue % Comments ON
\ifcomments

\definecolor{blueblack}{rgb}{0,0,.7}
\newcommand{\emphdef}[1]{%                                                      
	\textcolor{blueblack}{%                                                       
		\textbf{\emph{#1}}%                                                         
	}%                                                                            
}

\newcommand{\R}{\mathbb{R}}

\newcommand{\NP}{\textbf{NP}}

\newcommand{\G}{\mathcal{G}}
\newcommand{\Sh}{\mathcal{S}}
\newcommand{\ret}[2]{#1|^{#2}}

\DeclareMathOperator{\h}{height}

\newtheorem{theorem}{Theorem}
\newtheorem{lemma}[theorem]{Lemma}

\newtheorem{proposition}[theorem]{Proposition}

\title{On the complexity of optimal homotopies}%
\author{%
	Erin Wolf Chambers
	\thanks{Dept. of Computer Science, Saint Louis University. email: \protect\url{echambe5@slu.edu}}%
	\and%
	Arnaud de Mesmay%
	\thanks{Univ. Grenoble Alpes, CNRS, Grenoble INP, GIPSA-lab, 38000 Grenoble, France. email \protect\url{arnaud.de.Mesmay@gipsa-lab.fr}}%
	\and%
	Tim Ophelders%
	\thanks{Dept. of Mathematics and Computer Science, TU Eindhoven. email: \protect\url{t.a.e.ophelders@tue.nl}}%
}%

\date{}

\begin{document}
\begin{titlepage}

\maketitle\thispagestyle{empty}

\begin{abstract}\normalsize
	In this article, we provide new structural results and algorithms
	for the \textsc{Homotopy Height} problem. In broad terms, this
	problem quantifies how much a curve on a surface needs to be
	stretched to sweep continuously between two positions. More
	precisely, given two homotopic curves $\gamma_1$ and $\gamma_2$ on a
	combinatorial (say, triangulated) surface, we investigate the
	problem of computing a homotopy between $\gamma_1$ and $\gamma_2$
	where the length of the longest intermediate curve is
	minimized. Such optimal homotopies are relevant for a wide range of
	purposes, from very theoretical questions in quantitative homotopy
	theory to more practical applications such as similarity measures on
	meshes and graph searching problems.

	We prove that \textsc{Homotopy Height} is in the complexity class
	\textbf{NP}, and the corresponding exponential algorithm is the best
	one known for this problem. This result builds on a structural
	theorem on monotonicity of optimal homotopies, which is proved in a
	companion paper. Then we show that this problem encompasses the
	\textsc{Homotopic Fr\'echet distance} problem which we therefore
	also establish to be in \textbf{NP}, answering a question which has previously been
	considered in several different settings. We also provide an $O(\log
	n)$-approximation algorithm for \textsc{Homotopy Height} on surfaces
	by adapting an earlier algorithm of Har-Peled, Nayyeri, Salvatipour
	and Sidiropoulos in the planar setting.

\end{abstract}
\end{titlepage}
\setcounter{page}{1}

\section{Introduction}

This paper considers computational questions pertaining to
\textit{homotopies}: in broad terms, a homotopy between two curves in
a topological space is a continuous deformation between these two
curves. This can be formalized either in a continuous setting, where
it constitutes one of the fundamental constructs of algebraic topology,
but also in a more discrete one, where the input is a simplicial, or
more generally cellular description of a topological space; this
latter setting will be the focus of this article. While considerably more
restrictive than the more traditional mathematical settings, this setting is nonetheless
of key importance in applications areas such as graphics or medical
imaging, where inputs are generally represented by triangular meshes
built upon scanned point sets from an underlying 3D object.

Investigating homotopies from a computational perspective is a
well-studied problem, dating back to the work of Dehn~\cite{d-tkzf-12}
on contractibility of curves, which has strong ties to geometric group
theory. While deciding whether two curves in a $2$-dimensional complex
are homotopic is well-known to be undecidable in general (see for
example Stillwell~\cite{s-ctcgt-80}), when the underlying space is a
surface, efficient, linear-time algorithms have been designed to test
homotopy~\cite{dg-tcs-99,lr-hts-12,ew-tcsr-13}. In this article, we
add a quantitative twist to this problem: the \textsc{Homotopy Height}
problem consists, starting with two disjoint homotopic curves on a
combinatorial surface, of finding the homotopy of minimal height, that is,
where the length of the longest intermediate curve in the homotopy is
minimized.   (We refer the reader to Section~\ref{S:preliminaries} for formal 
definitions.)  The notion of homotopy height has
obvious appeal from a practical perspective, as it quantifies how long
a curve has to be to overcome a hurdle: for example, deciding whether
a bracelet is long enough to slide off over a hand without breaking is essentially the 
question of homotopy height. 
From a computational side, deformations of minimal height
minimize the necessary stretch and  can be used to quantify how
similar curves are, as in map or trajectory analysis.

\subsection{Our results}

We begin by considering two curves forming the boundary of a discrete
annulus, and study the homotopy between these boundaries of minimal
height. Our article leverages on recent results in Riemannian
geometry~\cite{cl-chidhh-14,cr-mhcdrs-16}, and in particular on a
companion article co-authored with Gregory Chambers and Regina
Rotman~\cite{ccmor-mcbd-17} where we prove
that in the Riemannian setting, such an optimal homotopy can be
assumed to be very well behaved. Firstly, it can be assumed to
be an isotopy, so that all the intermediate curves remain
simple. Secondly, this isotopy can be assumed to only move in one
direction and never sweeps any portion of the disk twice; we refer 
to this property as \textbf{monotonicity}, which we will define more precisely in
Section~\ref{S:continuous_discrete}.

These isotopy and monotonicity properties turn out to be a key
ingredient for computational purposes, once we translate those results
to the discretized setting.  First, via some surgery arguments, it
allows us to prove that \textsc{Homotopy Height} is in \textbf{NP}
(Theorem~\ref{T:NP}). The corresponding exponential time algorithm is
to our knowledge the best exact algorithm for \textsc{Homotopy
	Height}. We note that our setting is very general, and also implies
\textbf{NP}-membership for a variant of \textsc{Homotopy Height} in a
more restricted setting that was considered in earlier
papers~\cite{Brightwell09submodularpercolation,homotopyheight,hnss-hwdmml-16},
as well as for 
\textsc{Homotopic Fréchet distance}, where this was
still open despite the recent articles investigating this
distance~\cite{ccellt-hfdbcw-10,hnss-hwdmml-16}.
Then, further surgery arguments allow us to provide an $O(\log
n)$-approximation algorithm for \textsc{Homotopy Height}
(Corollary~\ref{C:logn}), by relying on an earlier $O(\log n)$
approximation-algorithm of Har-Peled, Nayyeri, Salvatipour and
Sidiropoulos~\cite{hnss-hwdmml-16} for homotopy height in a more
restricted setting. Finally, we show that monotonicity directly
implies an equivalence between the \textsc{Homotopy Height} problem
and a seemingly unrelated graph drawing problem which we call
\textsc{Minimal Height Linear Layout}. Therefore, we obtain
that this problem is also in \textbf{NP} and we provide an $O(\log n)$
approximation for it.

\subsection{Related work}

Optimal homotopies (for several definitions of optimal) have been studied extensively in the mathematical community, mainly in Riemannian settings. This literature fits broadly in the setting of quantitative homotopy theory, initially introduced by Gromov~\cite{g-qht-99}, which aims at introducing a quantitative lens in the study of topological invariants on manifolds. Probably the most extensively considered notion of optimality is the study of homotopies minimizing the area swept; see for example~\cite{lawson1980lectures} for an overview of some variants of this problem, or~\cite{white1984} for a discussion of how minimum area homotopies and homologies are connected in higher dimensions.  The notion of controlling the width of a homotopy has also been studied~\cite{width1980,geodesicwidth2013}, and more recent work on minimal height homotopies~\cite{cl-chidhh-14,cr-mhcdrs-16} laid the foundation for our results in this work.  

\begin{figure}\centering
\includegraphics{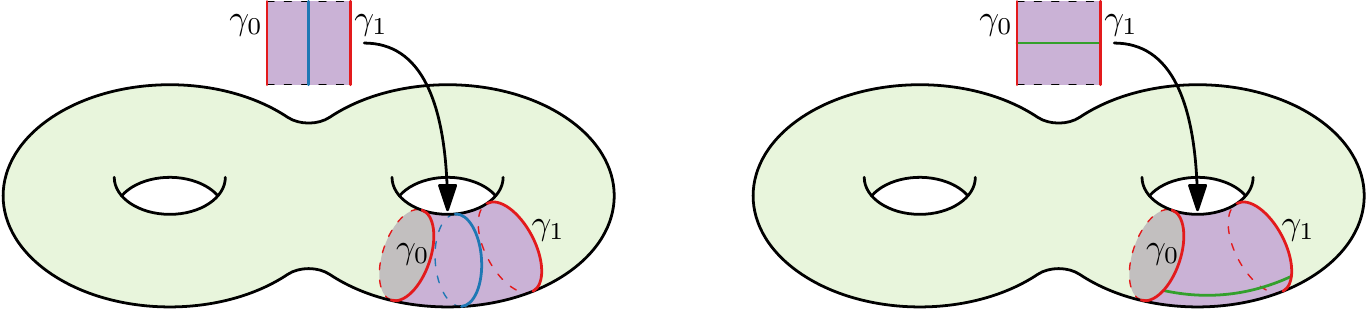}
\caption{Left: the height of a homotopy between homotopic curves~$\gamma_0$ and~$\gamma_1$ measures the maximum amount an intermediate curve must stretch during the homotopy. Homotopies minimizing this amount of stretch measure the homotopy height.
Right: the width of a homotopy measures the maximum length of a ``slice" of the homotopy connecting the two boundary curves. Homotopies minimizing the length of this slice measure the homotopy width, also known as the homotopic Fr\'echet distance.}
\label{fig:widthandheight}
\end{figure}

On the computational side, the rise of Fr\'echet distance for measuring similarity between curves was a prime motivation for the notion of comparing two curves; see for example~\cite{alt2009} for a survey.  Generalizing  Fr\'echet distance to surfaces led to the homotopic  Fr\'echet distance, which is essentially the same as finding a minimum width homotopy given two input cycles on a surface; algorithms are known to calculate this in polynomial time for two input curves in the plane minus a set of obstacles~\cite{CVE08} and to approximate this in the discrete settings where the two curves bound a disk~\cite{hnss-hwdmml-16}.  

More directly, minimum height homotopies have been studied from the computational perspective in various discretized settings~\cite{homotopyheight,hnss-hwdmml-16}, although mainly to discuss the complexity of the problem.  Indeed, as it was not known if the optimal height homotopy was even monotone, the complexity of the problem was completely open.  Since the monotonicity result also holds in more geometric settings~\cite{ccmor-mcbd-17}, a recent paper also examined one natural geometric setting, where the goal is to morph across a polygonal domain in Euclidean space with point obstacles; this work presents a lower bound that is linear in the number of obstacles, as well as a 2-approximation for the arbitrary weight obstacles and an exact polynomial time algorithm when all obstacles have unit weight~\cite{esa-hh2017}.
The same problem also arises as a combinatorics question in lattice theory as a \emph{b-northward migration}, where the authors leave monotonicity of such migrations as an open question~\cite{Brightwell09submodularpercolation}.

\subsection{Relations to graph searching and width parameters}

This work also connects to sweep and search parameters in graph theory; see for example~\cite{f-abggs-08} for a survey of this topic.  In each variant, the game consists of finding the minimum number of searchers needed, where the goal is to find or isolate a hidden fugitive.  For example, in the node searching variant,  the fugitive hides on edges, all of which are originally contaminated, and the searchers clear an edge if two are on its incident vertices.  In this variant, edges can be recontaminated if they are connected to a contaminated edge by a path without searchers, and the game ends when everything is decontaminated.

One key issue in these games is precisely that of monotonicity, or of determining whether in an optimal strategy, edges get recontaminated.  In the node searching variant, monotonicity was established by Lapaugh~\cite{l-rdnhsg-93}, and the argument was simplified by Bienstock and Seymour~\cite{bs-mgs-91}.  One important corollary to monotonicity for these games is that it immediately shows the problem lies in NP, since a strategy can be certified by the list of edges cleared.

Our homotopy problem is quite similar to these graph parameters; sweeping a disk while keeping the length small is intuitively quite similar to blocking in a fugitive.  While our  problem does display minor technical differences with the aforementioned variant -- most notably, our setting is naturally edge-weighted and the cost is measured on the edges and not the vertices -- the key difference is the one of \textit{connectedness}, as node-searching games may allow for disconnected strategies. An important variant of node searching, called \textit{connected node searching}, requires additionally that the decontaminated space remains connected, but makes no restriction on the uncontaminated space.

For graph searching problems, the main argument to establish monotonicity does not maintain connectivity~\cite{bs-mgs-91}, and it was proven that an optimal strategy for connected node searching may indeed be non-monotone~\cite{yda-sglcn-04}. By contrast, Theorem~\ref{T:monotonediscrete} establishes monotonicity of the optimal homotopy in our setting, and the arguments differ radically from the ones of Lapaugh and Bienstock and Seymour. As such, we identify in this paper a new variant of graph searching which is somewhat tractable (i.e., in \textbf{NP}) and introduce a new proof technique to establish monotonicity results.

Finally, when monotonicity is established, graph searching parameters
are very intimately related to width parameters of graphs.  Minimum
cut linear arrangement (also known as cut-width) is closely connected
to the \textsc{Minimum Height Linear Layout} problem, which we show to
be equivalent to \textsc{Homotopy Height}, but the key difference is
that it may break the embedding of the graph. Thus, NP-hardness
reductions for this problem~\cite{MONIEN1988209} do not imply hardness
for our problem.  Connected variants of various width parameters give
rise to \textit{connected pathwidth}~\cite{d-pcp-12} and
\textit{connected treewidth}~\cite{fn-ctcgs-06}, but in contrast to
our homotopies, these parameters are only connected ``on one side'',
which makes them incomparable. We believe that the
``doubly-connected'' aspect of homotopy height makes it a worthwhile
new graph parameter which could offer insights to other parameters in
this area.

\paragraph*{Outline of the paper.} After introducing the preliminaries in Section~\ref{S:preliminaries}, we lay the foundations of this work by explaining the structural theorems we rely on in Section~\ref{S:theorems}.
In Section~\ref{S:lemmas} we establish surgery lemmas based on the idea of \textit{retractions}.
Then, in Section~\ref{S:NP} we prove that \textsc{Homotopy Height} is in \textbf{NP}.
In Section~\ref{S:variants} we draw connections with \textsc{Homotopic Fr\'echet Distance}, and we leverage on these connections to provide an $O(\log n)$-approximation algorithm for \textsc{Homotopy Height}.

%%%%%%%%%%%%%%%%%%
\section{Preliminaries}\label{S:preliminaries}

\paragraph{Homotopy and Isotopy.}
Let~$\Sigma$ be a surface, endowed with a cellularly embedded graph $G$ with $n$ vertices such as in Figure~\ref{F:G}, and let~$\gamma_0$ and~$\gamma_1$ be two simple cycles on~$G$ bounding an annulus.
	\begin{wrapfigure}[12]{r}{0pt}\centering%
		\includegraphics{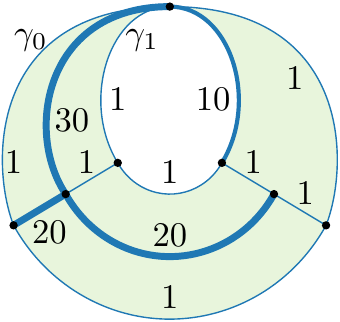}%
		\caption{Example instance~$G$, based on an example in~\cite{Brightwell09submodularpercolation}.}%
		\label{F:G}%
	\end{wrapfigure}
	A \emphdef{discrete homotopy}~$h$ between~$\gamma_0$ and~$\gamma_1$ is a sequence of cycles~$h(t_i)$ with~$0=t_0\leq\dots\leq t_i\leq\dots\leq t_m=1$, with~$h(t_0)=\gamma_0$ and~$h(t_1)=\gamma_1$ and any two consecutive paths~$h(t_i)$ and~$h(t_{i+1})$ are one \emph{move} apart. The intermediate curves $h(t)$ are called \emphdef{level curves} or \emphdef{intermediate curves}.
	A move is either a face-flip, an edge-spike or an edge-unspike (flip, spike or unspike, for short).
	A \emphdef{face-flip} for a face~$F$ replaces a single subpath~$p$ of~$h(t_i)\cap\partial F$ with the path~$\partial F\setminus p$ in~$h(t_{i+1})$.
	An \emphdef{edge-spike} for an edge~$u\to v$ replaces a single occurrence of a vertex~$u\in h(t_i)$ by the path~$u\to v\to u$ consisting of two mirrored copies of that edge in~$h(t_{i+1})$.
	Symmetrically, an \emphdef{edge-unspike} replaces a path~$u\to v\to u$ of~$h(t_i)$ by the single vertex~$u$ in~$h(t_{i+1})$.
	The~\emphdef{length}~$\ell(h(i))$ of a path~$h(i)$ is the sum of the weights of its edges (with multiplicity).
	The~\emphdef{height} of a homotopy~$h$ is the length of the longest path~$h(t_i)$.
	An~\emphdef{optimal homotopy} is one that minimizes the height.
	The~\emphdef{homotopy height} between $\gamma_0$ and $\gamma_1$ is the height of an optimal homotopy between $\gamma_0$ and $\gamma_1$.
	Figure~\ref{F:Gh} illustrates an optimal homotopy that uses only face-flips for the instance of Figure~\ref{F:G}.

	\begin{figure}[ht]\centering%
		\includegraphics{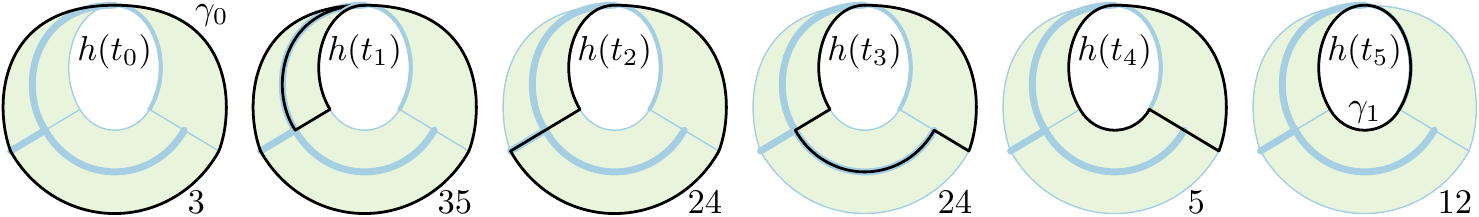}%
		\caption{An optimal homotopy~$h$ of height~$35$ for the instance of~Figure~\ref{F:G}.}%
		\label{F:Gh}%
		\vspace{-1em}
	\end{figure}

\noindent\paragraph{Cross Metric Surfaces.}
	\begin{wrapfigure}[12]{r}{0pt}\centering%
		\includegraphics{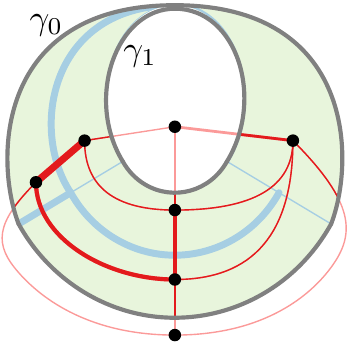}%
		\caption{Dual representation of Figure~\ref{F:G}.}%
		\label{F:dualG}%
	\end{wrapfigure}
		For most purposes, it is more convenient to think of this discrete model in a dual way, relying on the \textit{cross-metric surfaces}~\cite{ce-tnpcs-10} which are becoming increasingly used in the computational geometry and topology literature.
		In this dual setting, a cross-metric surface is a surface $\Sigma$ endowed with a weighted (dual) graph~$G^*$.

		Assuming the primal surface is connected, we obtain this dual graph by gluing a disk to each boundary component, taking the dual graph, and puncturing the vertices corresponding to the added disks, without removing the adjacent edges.
		Such that these (dual) edges end at the boundary of the cross-metric surface instead of at a vertex, see Figure~\ref{F:dualG}.

		For a curve~$\gamma$ on~$\Sigma$ with a finite number of crossings with~$G^*$, its length~$\ell(\gamma)$ is the weighted sum of the crossings~$\gamma \cap G^*$.
		Now, a homotopy between $\gamma_0$ and $\gamma_1$ is a homotopy in the usual sense, that is, 
		a continuous map $h \colon S^1 \times [0,1] \rightarrow \Sigma$ such that $h(\cdot,0)=\gamma_0$ and $h(\cdot,1)=\gamma_1$, 
		except that we require that the values of $t$ for which $h(\cdot,t)$ is not in general position with $G^*$ are isolated, 
		and each such curve has at most one such degeneracy%
		\footnote{Any homotopy can be made so by a small perturbation without increasing the height, so we always consider this hypothesis fulfilled in the remainder of the article.}%
		~$h(x,t)$ with~$G^*$.
		As before, the height of a homotopy is defined as the maximal length of an intermediate non-degenerate level curve~$h(t)$.
		A homotopy is an \textit{isotopy} if all the intermediate curves are simple.
		
	\begin{wrapfigure}{r}{0pt}\centering%
		\includegraphics{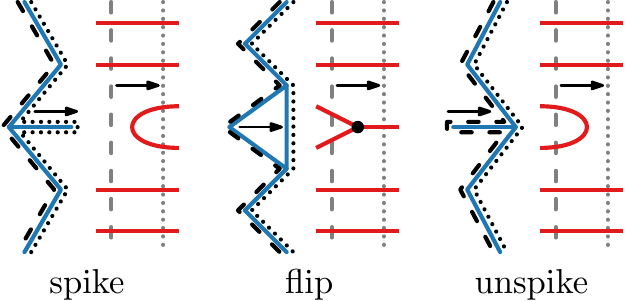}%
		\caption{Three moves in the primal (left) and dual (right) representation.}%
		\label{F:dualMoves}%
	\end{wrapfigure}
		Given a homotopy~$h^*$ in this setting, we obtain a discrete homotopy~$h$ on the primal graph~$G$ on~$\Sigma$ as follows.
		Pick a curve~$h^*(t_i)$ in each maximal interval of non-degenerate curves in~$h^*$ (all curves in such interval have the same crossing pattern with~$G^*$, and therefore the same length).
		Let~$h(t_i)$ be the curve on~$G$ whose sequence of vertices and edges corresponds to the sequence of faces and edges of~$G^*$ visited by~$h^*(t_i)$.
		This model is dual to the previous one, and Figure~\ref{F:dualMoves} illustrates how any move (flip, spike or unspike) connects two intermediate curves~$h(t_i)$ and~$h(t_{i+1})$.
		We say a discrete homotopy is an isotopy if it can be obtained from an isotopy in the dual setting.

%%%%%%%%%%%%%%%%%
\section{Isotopies and monotonicity of optimal homotopies}\label{S:theorems}
\label{S:continuous_discrete}

We begin by restating and explaining the two structural results that we will rely on. Introducing the relevant Riemannian background lies outside of the scope of this paper, so we will simply advise the uninitiated reader to picture a Riemannian surface as a surface embedded into $\mathbb{R}^3$, where the metric on the surface is the one induced by the usual Euclidean metric of $\mathbb{R}^3$. Thanks to the Nash-Kuiper embedding theorem (see~\cite{hh-ierme-06}), this naive idea looses no generality. We refer to standard textbooks on the subject for more proper background on Riemannian geometry, for example do Carmo~\cite{c-rg-92}.

The first theorem shows that up to an arbitrarily small additive factor, the homotopy of minimal height between two simple closed curves can be assumed to be an isotopy.

\begin{theorem}[{\cite[Theorem~1.1]{cl-chidhh-14}}]\label{T:isotopy}
	Let $\Sigma$ be two-dimensional Riemannian manifold with or without boundary, and let $\gamma_0$ and $\gamma_1$ be two non-contractible simple closed curves which are homotopic through curves bounded in length by $L$ via a homotopy $\gamma$.

	Then for any $\varepsilon>0$, there is an isotopy $\widetilde{\gamma}$ from $\gamma_0$ to $\gamma_1$ through curves of length at most $L+\varepsilon$.
\end{theorem}

\noindent\textbf{Remark.} The non-contractibility hypothesis is required because if $M$ is not a sphere, contractible cycles with opposite orientations are homotopic but not isotopic. However, if we disregard the orientations, the result holds in full generality.
\smallskip

This theorem has the following discrete analogue:

\begin{restatable}{theorem}{Tisotopydiscrete}\label{T:isotopydiscrete}
	Let $(\Sigma,G^*)$ be a cross-metric surface, and let $\gamma_0$ and $\gamma_1$ be two non-contractible simple closed curves on $(\Sigma,G^*)$ which are homotopic through curves bounded in length by $L$ via a homotopy $\gamma$.

	Then there is an isotopy $\widetilde{\gamma}$ from $\gamma_0$ to $\gamma_1$ through curves of length at most $L$.
\end{restatable}

The proof is exactly the same as the one of Theorem~\ref{T:isotopy}, except that it does not need the $\varepsilon$-slack: this was required to slightly perturb the curves so that they are simple but in the discrete setting it can be done with no overhead.

The second theorem shows that, \textit{when the starting and finishing curves of a homotopy are the boundaries of the manifold}, there exists an optimal homotopy that is monotone, i.e., that never backtracks, once again up to an arbitrarily small additive factor. Formally, if $\gamma$ is an isotopy (which we can assume the optimal homotopy to be, by Theorem~\ref{T:isotopy}) between $\gamma_0$ and $\gamma_1$, for $0\leq t\leq 1$, the curves $\gamma_t$ and $\gamma_1$ bound an annulus $A_t$. Then the isotopy $\gamma$ is \emphdef{monotone} if for $t<t'<1$, $\gamma_{t'}$ is contained in $A_t$.

\begin{theorem}[{\cite[Theorem~1.2 and the following paragraph]{ccmor-mcbd-17}}]\label{T:monotone}
	Let $M$ be a Riemannian annulus with boundaries $\gamma_0$ and $\gamma_1$ such that there exists a homotopy between $\gamma_0$ and $\gamma_1$ of height less than $L$.

	Then there exists a \textit{monotone} homotopy between $\gamma_0$ and $\gamma_1$ of height less than $L$.
\end{theorem}

Note that the $\varepsilon$-slack of Theorem~\ref{T:isotopy} is also present here but is hidden in the open upper bound of $L$. In this theorem, as was observed by Chambers and Rotman~\cite{cr-mhcdrs-16}, crediting Liokumovitch, the hypothesis that the manifold is entirely comprised between both curves is necessary: see~\cite[Figure~5]{cr-mhcdrs-16} for a counter-example.

In the discrete setting, the corresponding result is the following, where the definition of monotonicity is the same:

\begin{restatable}{theorem}{Tmonotonediscrete}\label{T:monotonediscrete}
	Let $(\Sigma,G^*)$ be a cross-metric annulus with boundaries $\gamma_0$ and $\gamma_1$ such that there exists a homotopy between $\gamma_0$ and $\gamma_1$ of height $L$.

	Then there exists a \textit{monotone} isotopy between $\gamma_0$ and $\gamma_1$ of height $L$.
\end{restatable}

The proof is exactly identical to the one in the Riemannian
settting and it yields a slightly stronger result, since the cross-metric setting removes
the need for perturbations and thus the need of an
$\varepsilon$-slack.

\textbf{Remark:} Let us observe that the discrete theorems are in some way more general than the Riemannian ones: not only do they bypass the need for some $\varepsilon$-slack, but they also directly imply their Riemannian converses by the following reduction. Starting with a Riemannian surface, and a (non-monotone) isotopy between two disjoint curves, one can find a triangulation of the surface allowing, at an $\varepsilon$-cost, to approximate the isotopy using only elementary moves. Then, after making this isotopy monotone in the discrete setting, one can translate it back into a monotone isotopy in the Riemannian setting by interpolating between the face and edge moves.

\section{Retractions and pausing at short cycles}\label{S:lemmas}

	In this section, we establish several technical lemmas which are necessary for our proofs in the next section.
	For simple closed curves~$\beta$ and~$\gamma$ bounding an annulus, denote that annulus by~$A(\beta,\gamma)$.
	Let~$\Sh(\beta,\gamma)$ be the set of closed curves in~$A(\beta,\gamma)$ homotopic to boundaries~$\beta$ and~$\gamma$, that do not intersect homotopic curves of shorter length.
	Then, for any point~$p\in\alpha\in\Sh(\beta,\gamma)$,~$\alpha$ is a shortest closed path through~$p$ in its homotopy class.
	Let~$\G(\beta,\gamma)$ be the set of minimum length simple closed curves homotopic to the boundaries of~$A(\beta,\gamma)$, then~$\G(\beta,\gamma)\subseteq\Sh(\beta,\gamma)$.

	We now introduce the concept of a retraction of a homotopy,
				which gives a way to shortcut a homotopy at a given curve,
				provided it is a curve of $\Sh(\beta,\gamma)$. This idea is implicit in Chambers and Rotman~\cite[Proof of
					Theorem~0.7]{cr-mhcdrs-16}, and we refer to their article
				for more details. For a monotone isotopy~$h$ between
				boundaries of an annulus~$A$, and a homotopic
				annulus~$A'\subset A$, define the
				\emphdef{retraction}~$\ret{h}{A'}(t)$ of~$h(t)$ to~$A'$ as the
				same curve with each arc of~$h(t)\setminus A'$ replaced by the
				shortest homotopic path along the boundary of~$A'$.  Although
				paths along~$\partial A'$ (dis)appear discontinuously as~$t$
				varies,~$\ret{h}{A'}$ can be obtained in the form of a
				discrete homotopy by (un)spiking these paths as they
				(dis)appear.  The resulting homotopy~$\ret{h}{A'}$ is a
				monotone isotopy.

	\begin{lemma}\label{L:semiGeodesic}
		If~$\alpha\in\Sh(\alpha,\gamma)$ and~$A(\alpha,\gamma)\subseteq A(\beta,\gamma)$, and~$h$ is a monotone isotopy from~$\beta$ to~$\gamma$ of height~$L$, then~$\ret{h}{A(\beta,\alpha)}$ is a monotone isotopy from~$\beta$ to~$\alpha$ with height at most~$L$.
	\end{lemma}
	\begin{proof}
		The retraction~$h'=\ret{h}{A(\beta,\alpha)}$ is a monotone isotopy from~$h'(0)=\beta$ to~$h'(1)=\alpha$.
		Let~$t'$ be the maximum~$t$ for which~$h(t)$ intersects~$A(\beta,\alpha)$.
		For~$t\geq t'$, we have~$h'(t)=\alpha$ and therefore~$|h'(t)|=|\alpha|\leq|h(t')|\leq L$.
		For~$t\leq t'$, each arc~$a$ of~$h(t)\setminus A(\beta,\alpha)$ is replaced in~$h'(t)$ by a homotopic path~$b$ along~$\alpha$ with~$|b|\leq|a|$, and thus~$|h'(t)|\leq|h(t)|\leq L$.
		Hence~$\h(h')\leq L$.
	\end{proof}

	\begin{lemma}\label{L:geodesic}
		If~$\alpha\in\Sh(\beta,\gamma)$, and~$h$ is a monotone isotopy from~$\beta$ to~$\gamma$ of height~$L$, then there is a monotone isotopy from~$\beta$ to~$\gamma$ of height at most~$L$ having~$\alpha$ as a level curve.
	\end{lemma}
	\begin{proof}
		We have~$\alpha\in\Sh(\alpha,\beta)$ and~$\alpha\in\Sh(\alpha,\gamma)$.
		So by Lemma~\ref{L:semiGeodesic}, the monotone isotopies~$\ret{h}{A(\beta,\alpha)}$ from~$\beta$ to~$\alpha$ and~$\ret{h}{A(\alpha,\gamma)}$ from~$\alpha$ to~$\gamma$ have height at most~$L$ and can be composed to obtain a monotone isotopy from~$\beta$ to~$\gamma$ of height at most~$L$ with~$\alpha$ as a level curve.
	\end{proof}

	\begin{lemma}\label{L:spanning}
		Let~$\Pi=\{\pi_1,\dots,\pi_m\}$ be a set of paths from~$\gamma_0$ to~$\gamma_1$ without proper pairwise intersections, where each~$\pi_i$ is a shortest homotopic path in~$A(\gamma_0,\gamma_1)$ between its endpoints.
		If~$h$ is a monotone isotopy from~$\gamma_0$ to~$\gamma_1$ of height~$L$, then there exists a monotone isotopy of height at most~$L$ whose level curves all cross each~$\pi_i$ at most once (after infinitesimal perturbations).
	\end{lemma}
	\begin{proof}
		Denote by~$c(a,b)$ the number of proper intersections of curves~$a$ and~$b$, and by~$c_\Pi(a)=\sum_{\pi\in \Pi} c(a,\pi)$ the total number of intersections of~$a$ with~$\Pi$.
		Let~$C_h=\max_t c_\Pi(h(t))$ be the maximum total number of intersections over all~$t$, and let~$I_h$ be the set of maximal intervals~$(\tau_0,\tau_1)$ with~$c_\Pi(h(t))=C_h$ if~$t\in(\tau_0,\tau_1)\in I_h$. If~$C_h=m$, each level curve of~$h$ crosses each~$\pi_i$ exactly once and we are done, thus we assume in the following that~$c_\Pi(h(0))=c_\Pi(h(1))=m<C_h$.
		
		If~$C_h>m$, we obtain a homotopy~$h'$ from~$h$ with~$C_{h'}<C_h$ by, for each interval~$(\tau_0,\tau_1)\in I_h$, replacing subhomotopy~$h|_{(\tau_0,\tau_1)}$ of~$h$ by some~$h^*=h'|_{(\tau_0,\tau_1)}$ with~$C_{h^*}<C_h$.

		Consider a single interval~$(\tau_0,\tau_1)\in I_h$ and let~$A=A(h(\tau_0),h(\tau_1))$.
		Then~$\Pi\cap A$ consists of~$C_h$ subarcs of~$\Pi$, each connecting the two boundaries of~$A$.              
		For~$t\in(\tau_0,\tau_1)$,~$h(t)$ intersects each such arc exactly once, and each~$h(t)$ intersects these arcs in the same order.
		Among the components of~$A\setminus \Pi$, there is a disk~$D_0$ bounded by one arc of~$h(\tau_0)$ and two arcs of~$\pi_i\cap A$, and a disk~$D_1$ bounded by one arc of~$h(\tau_1)$ and one arc of~$\pi_j$, such that these disks contain no other arcs of~$\Pi$.
		\begin{figure}[h]\centering%
			\includegraphics{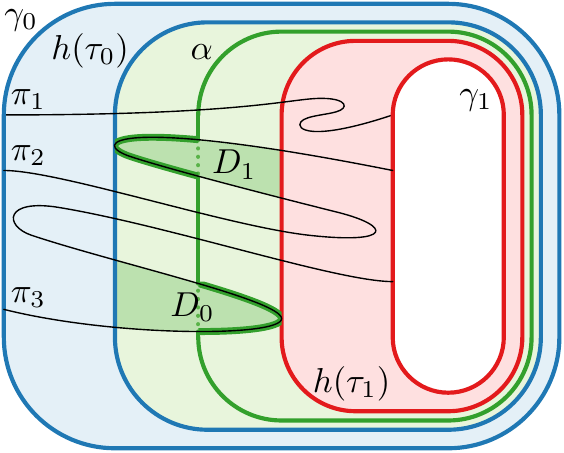}
			\caption{Choosing~$\alpha$ such that~$C_{h^*}<C_h$.}
			\label{F:spanning}
		\end{figure}%
		
		We can find~$\alpha\in\G(h(\tau_0),h(\tau_1))$ that intersects any arc of~$A\cap \Pi$ at most once (in the same order as~$h(t)$), and does not intersect the interiors of~$D_0$ and~$D_1$ (because the two arcs of~$\Pi$ on their boundary form a shortest path).
		Then~$c_\Pi(\alpha)<C_h$ and the retraction~$h_0=\ret{h}{A(h(\tau_0),\alpha)}$ has~$C_{h_0}<C_h$, since any arc~$h_0(t)$ has fewer intersections than~$h(t)$ has with~$\Pi$ (in particular with the boundary of~$D_1$).
		Symmetrically, for~$h_1=\ret{h}{A(\alpha,h(\tau_1))}$ we have~$C_{h_1}<C_h$.
		Since the composition~$h^*=h_0 h_1$ is a homotopy from~$h(\tau_0)$ to~$h(\tau_1)$ with~$C_{h^*}<C_h$ and height at most~$L$ (by Lemma~\ref{L:geodesic}), we can use this as a replacement for~$h|_{(\tau_0,\tau_1)}$ in~$h'$.
		By induction, we obtain a homotopy of height at most~$L$ whose level curves all cross each~$\pi_i$ at most once.
	\end{proof}

\section{Computing homotopy height in NP}\label{S:NP}
	In this section, we show that in the discrete setting, there is an optimal homotopy with a polynomial number of moves.
	First, we show that there is a homotopy that flips each face exactly once.

	\begin{lemma}
	\label{L:polynomialfaceflips}
	For an annulus~$(\Sigma,G)$ bounded by~$\gamma_0$ and~$\gamma_1$, there is a homotopy of minimum height between $\gamma_0$ and $\gamma_1$ that flips each face of~$G$ exactly once.
	\end{lemma}
	\begin{proof}
	By Theorem~\ref{T:monotonediscrete}, some homotopy~$h$ of minimum height is a monotone isotopy.
	For two consecutive level curves~$h(t)$ and~$h(t')$ in a monotone isotopy, the move between~$h(t)$ and~$h(t')$ flips face~$F$ if and only if~$F$ lies in~$A(h(t'),\gamma_1)$ or~$A(h(t),\gamma_1)$ but not both.
	Because~$A(h(0),\gamma_1)$ contains all faces, and~$A(h(1),\gamma_1)$ contains none, each face is flipped at least once.
	By monotonicity, we have for~$0\leq t'<t\leq 1$, that~$A(h(t'),\gamma_1)\supseteq A(h(t),\gamma_1)$.
	So, if face~$F$ does not lie in~$A(h(t),\gamma_1)$, it will not be flipped again in~$h|_{(t,1]}$.
	Hence each face is flipped exactly once.
	\end{proof}

	It remains to show that each edge is involved in a polynomial number of (un)spike moves; note that this does not directly follow from monotonicity, since a second spike of the same edge does not violate monotonicity (as can easily be seen in the dual setting).

\paragraph{Postponing spikes.}
	Before we bound the number of spike moves, we transform an optimal monotone isotopy~$h$ into one where each spike move is delayed as much as possible, and each unspike move happens as soon as possible.
	We explain this transformation in the dual setting.
	
	Suppose a spike move occurs for edge~$e$ between~$h(t_i)$ and~$h(t_{i+1})$, then denote by~$s$ the (unique) arc of~$A(h(t_i),h(t_{i+1}))\cap G^*$ both of whose endpoints lie on~$h(t_{i+1})$.
	This arc is a subarc of the dual edge~$e^*$.
	Consider the maximum~$j>i$, for which the component~$s_j$ of~$e^*\cap A(\gamma_0,h(t_j))$ containing~$s$ has both endpoints on~$h(t_j)$, and for all~$t_i<t\leq t_j$, curve~$h(t)$ has exactly two crossings with~$s_j$ (so the only action performed on arc~$s_j$ was the spike between~$h(t_i)$ and~$h(t_{i+1})$).
	Then~$s_j$ and~$h(t_j)$ enclose a disk~$D_j$.
	If the interior of~$D_j$ contains no edges of~$G^*$, we can delay the spike of~$e$ at least until just before~$t_j$, as illustrated in Figure~\ref{F:dualSimplify} (a), where~$D_j$ is shaded.

	Depending on what happens in the move between~$h(t_j)$ and~$h(t_{j+1})$, we may transform the isotopy further.
	This move is either (1) an unspike attached to~$s_j$, or (2) a face-flip connected to one endpoint or (3) both endpoints\footnote{This happens only if the primal edge is adjacent to only one face of~$G$.} of~$s_j$, or (4) a face-flip or spike inside~$D_{j+1}$.
	In cases (1) and (2), we cancel the spike against the unspike or flip, as illustrated in Figure~\ref{F:dualSimplify} (b) and (c).
	We do not postpone the spike in cases (3) and (4).
	Symmetrically, unspike moves can be made to happen earlier.
	Observe that these operations cannot increase the height of a homotopy since each level curve in the resulting homotopy crosses a subset of the edges of some curve in the original homotopy.
	\begin{figure}[h]\centering%
		\includegraphics{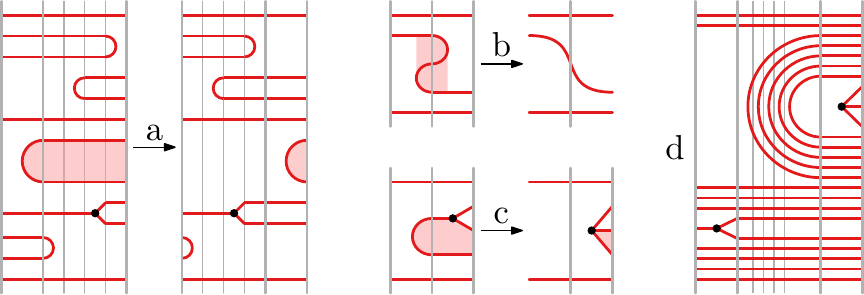}%
		\caption{Delaying spikes (a). %
		Canceling spikes against unspikes (b) or faces (c). %
		Part of a reduced isotopy (d).}%
		\label{F:dualSimplify}%
	\end{figure}

	Call a homotopy \emphdef{reduced} if it is the result of applying the above rules to~$h$ until no spike can be canceled or postponed until after a flip or unspike, and no unspike can be canceled or be made to happen before any prior flip or spike.
	Observe that starting from an optimal monotone isotopy, the reduced isotopy is also an optimal monotone isotopy.
	The structure of reduced homotopies is given in Lemma~\ref{L:reduced}.
	\begin{lemma}
		Between any two consecutive face-flips in a reduced isotopy lies a single (possibly empty) path of unspike moves followed by a (possibly empty) path of spiked moves.\label{L:reduced}
	\end{lemma}
	\begin{proof}
		In a reduced homotopy, no unspike follows a spike move, and any spikes that remain `surround' the next face-flip (if any), see Figure~\ref{F:dualSimplify} (d).
		Symmetrically, all unspikes between two consecutive face-flips surround the previous face-flip (if any).
		From the primal perspective, these unspike moves form a path from the previously flipped face and spike moves form a path towards the next flipped face.
	\end{proof}
	
	Any reduced homotopy starts with zero or more unspikes from~$\gamma_0$, after which a possibly empty path of spikes to the first face-flip occurs, then that face is flipped, and a possibly empty path of unspikes enabled by this flip occurs.
	Subsequently, a spiked path, face-flip, and unspiked path occur for the remaining faces.
	Finally, a sequence of spikes towards~$\gamma_1$ may occur.
	We may assume that on~$\gamma_0$ and~$\gamma_1$, any two consecutive edges are different, such that no immediate unspike moves are possible from~$\gamma_0$, and no immediate spike moves are possible to~$\gamma_1$.
	Otherwise we may by Lemma~\ref{L:semiGeodesic} perform those moves immediately without increasing the homotopy height.

\paragraph{Bounding spike moves.}
	
	We are now ready to bound the number of spike and unspike moves in an optimal homotopy.
	Call a homotopy~$h$~\emphdef{good} if it is a minimum-height reduced monotone isotopy and it has a minimum number of moves.
	By Theorems~\ref{T:isotopydiscrete} and~\ref{T:monotonediscrete} , the height of~$h$ is the homotopy height between~$\gamma_0$ and~$\gamma_1$.

	Define an edge-spike of an edge~$e$ to be \emphdef{between} existing copies of~$e$, if the portion of the dual edge~$e^*$ crossed by the (dual) level curve, lies between two existing crossings of the level curve with~$e^*$, such as in Figure~\ref{F:case1}.
	We show that such spikes never appear in~$h$.
	\begin{figure}[ht]\centering
		\includegraphics{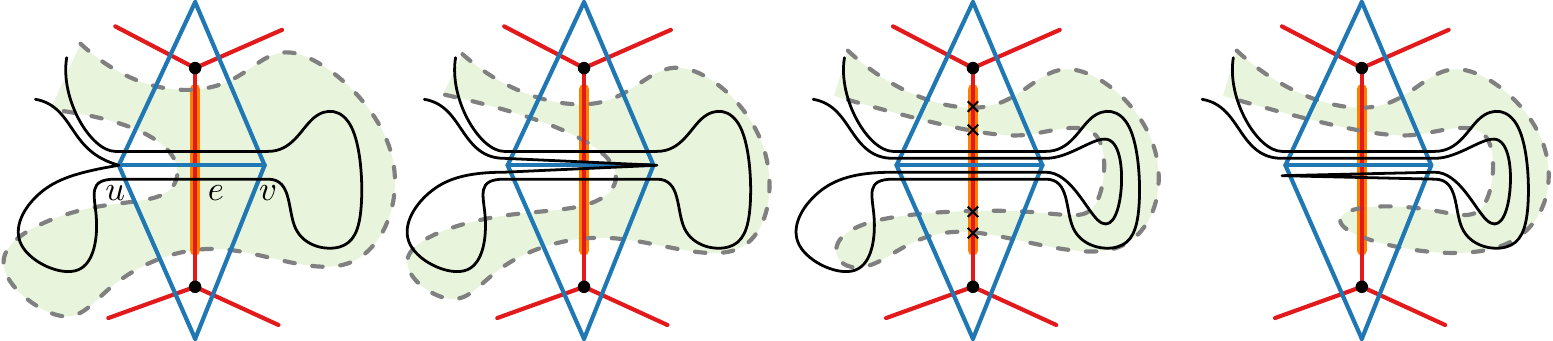}
		\caption{Development of a spike between existing copies of~$e$.
		Part of the graph in red (dual) and blue (primal) and the level curve in gray dashed (dual) and black (primal, perturbed).}
		\label{F:case1}
	\end{figure}

	\begin{lemma}\label{L:betweenSpikes}
		If homotopy~$h$ is good, there are no spikes between existing copies of any edge~$e$.
	\end{lemma}
	\begin{proof}
		Suppose the move from~$h(t_i)$ to~$h(t_{i+1})$ is the last move between existing copies of the same edge, and assume this move is a spike of edge~$e=(u,v)$ from~$u$ to~$v$.
		In the dual setting, consider the component~$\pi$ of~$e^*\cap A(h(t_i),\gamma_1)$ that is crossed by the spike move (highlighted in Figure~\ref{F:case1}).
		Let~$c(t)$ be the number of crossings of~$h(t)$ with~$\pi$, then for some~$\tau_0$ between~$t_i$ and~$t_{i+1}$,~$c(\tau_0)=3$, and for some unique~$\tau_1>\tau_0$,~$c(\tau_1)=3$ again, and for~$\tau_0<t<\tau_1$, we have~$c(t)=4$ (because we assumed this was the last spike between existing copies of any edge).

		For~$\tau_0<t<\tau_1$, label the four crossings of~$h(t)$ with~$\pi$ by~$p_1(t)$,~$p_2(t)$,~$p_3(t)$, and~$p_4(t)$, in order along~$e^*$, so the spike move at~$\tau_0$ creates~$p_2$ and~$p_3$.
		Consider the three components~$C_1(t)$,~$C_2(t)$ and~$C_3(t)$ of~$A(h(t),\gamma_1)\setminus\pi$, such that~$C_1$ touches~$p_1$ and~$p_2$ from the dual face of~$u$, and~$C_2$ touches~$p_3$ and~$p_4$ from the dual face of~$u$, and~$C_3$ touches~$e^*$ in two segments from the dual face of~$v$.
		Because~$C_3$ lies between~$C_1$ and~$C_2$,~$h$ will first contract either component~$C_1$ or~$C_2$, namely at~$h(\tau_1)$.
		Assume without loss of generality that~$C_2$ contracts first.

		We modify~$h|_{[\tau_0,\tau_1]}$ such that any level curve crosses~$\pi$ at most twice by reconnecting the neighborhood of~$\pi$, whose local structure evolves exactly as depicted in the top row of Figure~\ref{F:case1local}.
		We essentially remove crossings~$p_2$ and~$p_3$, and reconnect~$\partial C_1(t)\cap h(t)$ with~$\partial C_2(t)\cap h(t)$ using a (zero-length) segment along~$\pi$ in face~$u^*$.
		On the other side, consider the arc of~$\partial C_3(t)\cap h(t)\cap v^*$ with~$p_4(t)$ as endpoint.
		We cut this arc in two subarcs~$a$ and~$b$, where~$a$ has~$p_4(t)$ as endpoint, and connect the other endpoint to the arc of~$\partial C_3(t)\cap h(t)$ at the endpoint at~$p_2(t)$ using a segment along~$\pi$ in~$v^*$.
		Similarly, we connect the endpoint of that at~$p_3(t)$ to the loose end of~$b$.
		These reconnections are depicted in the bottom row of Figure~\ref{F:case1local}.
		A more global view (corresponding to Figure~\ref{F:case1}) is illustrated in Figure~\ref{F:case1solution}.

	\begin{figure}[ht]\centering
		\includegraphics{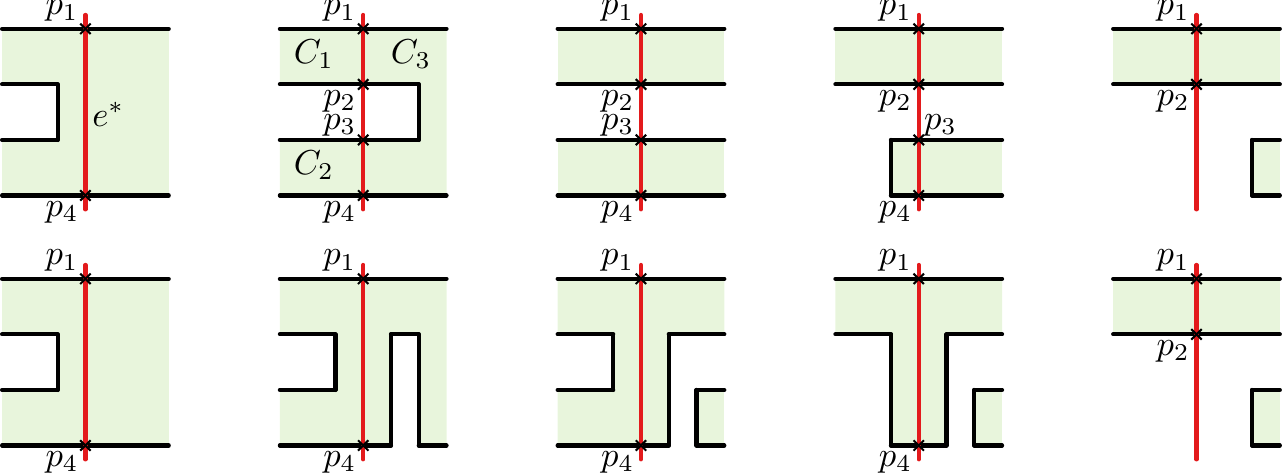}
		\caption{Top: the neighborhood of~$e^*$ throughout~$h$. Bottom: the reconnected homotopy, reducing crossings with~$e^*$. From left to right: the homotopy just before~$\tau_0$, just after~$\tau_0$, between~$\tau_0$ and~$\tau_1$, just before~$\tau_1$, and just after~$\tau_1$.}
		\label{F:case1local}
	\end{figure}

	\begin{figure}[ht]\centering
		\includegraphics{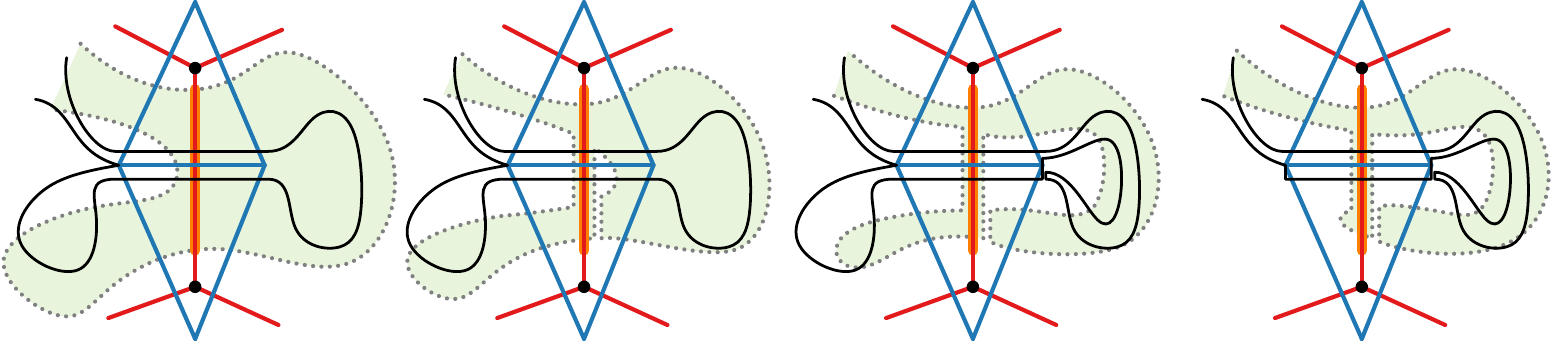}
		\caption{Figure~\ref{F:case1} after a local surgery that avoids the spike between copies of~$e$.}
		\label{F:case1solution}
	\end{figure}

		Observe that the reconnected curves can be made to appear continuously in such a way that they form a monotone isotopy.
		Because level curves only changed in the neighborhood of~$\pi$, where they were shortened by avoiding the crossings with~$\pi$, we have an isotopy whose height is at most that of~$h$, and in which at least one spike is removed.
		So, because~$h$ was optimal, we have constructed an optimal monotone isotopy with fewer moves.
		Therefore, the corresponding reduced isotopy also has fewer moves, contradicting that~$h$ was good.
	\end{proof}

	Our final step towards bounding the number of edge spikes is to derive a contradiction if for some interval~$[\tau_0,\tau_1]$ without face-flips, an edge~$e$ is spiked (or unspiked)~$5$ times in~$h|_{[\tau_0,\tau_1]}$. The proof is similar to that of Lemma~\ref{L:betweenSpikes}.
	\begin{lemma}\label{L:shortSpikedPaths}
		For a good homotopy~$h$, any subhomotopy~$h|_{[\tau_0,\tau_1]}$ contains either a face-flip, or at most~$4$ spike (and at most~$4$ unspike) moves of the same edge.
	\end{lemma}
	\begin{proof}
		Suppose~$h|_{[\tau_0,\tau_1]}$ contains no face-flip, then because~$h$ is reduced, the spike moves in~$h|_{[\tau_0,\tau_1]}$ form a path~$\sigma$ of spike moves in~$G$.
		Assume for a contradiction that some edge~$e=(u,v)$ lies on~$\sigma$ at least~$5$ times.
		We say two spikes~$s_1$ and~$s_2$ are consecutive on~$e^*$ if no spike occurs on the arc of~$e^*$ between the first crossing of~$s_1$ with~$e^*$ and the first crossing of~$s_2$ with~$e^*$.

		Because by Lemma~\ref{L:betweenSpikes},~$h$ does not contain spikes between existing copies of edges, we can find three spikes~$s_1$,~$s_2$ and~$s_3$ of~$e$ on~$\sigma$ where~$s_1$ and~$s_2$ as well as~$s_2$ and~$s_3$ are consecutive on~$e^*$, and~$s_1$ happens before~$s_2$ and~$s_2$ happens before~$s_3$.
		Let~$\sigma_0$,~$\sigma_1$,~$\sigma_2$ and~$\sigma_3$ be the subpaths of~$\sigma$ such that~$\sigma=\sigma_0 s_1 \sigma_1 s_2 \sigma_2 s_3 \sigma_3$, also labeled in Figure~\ref{F:case2}.
		
	\begin{figure}[h]\centering
		\includegraphics{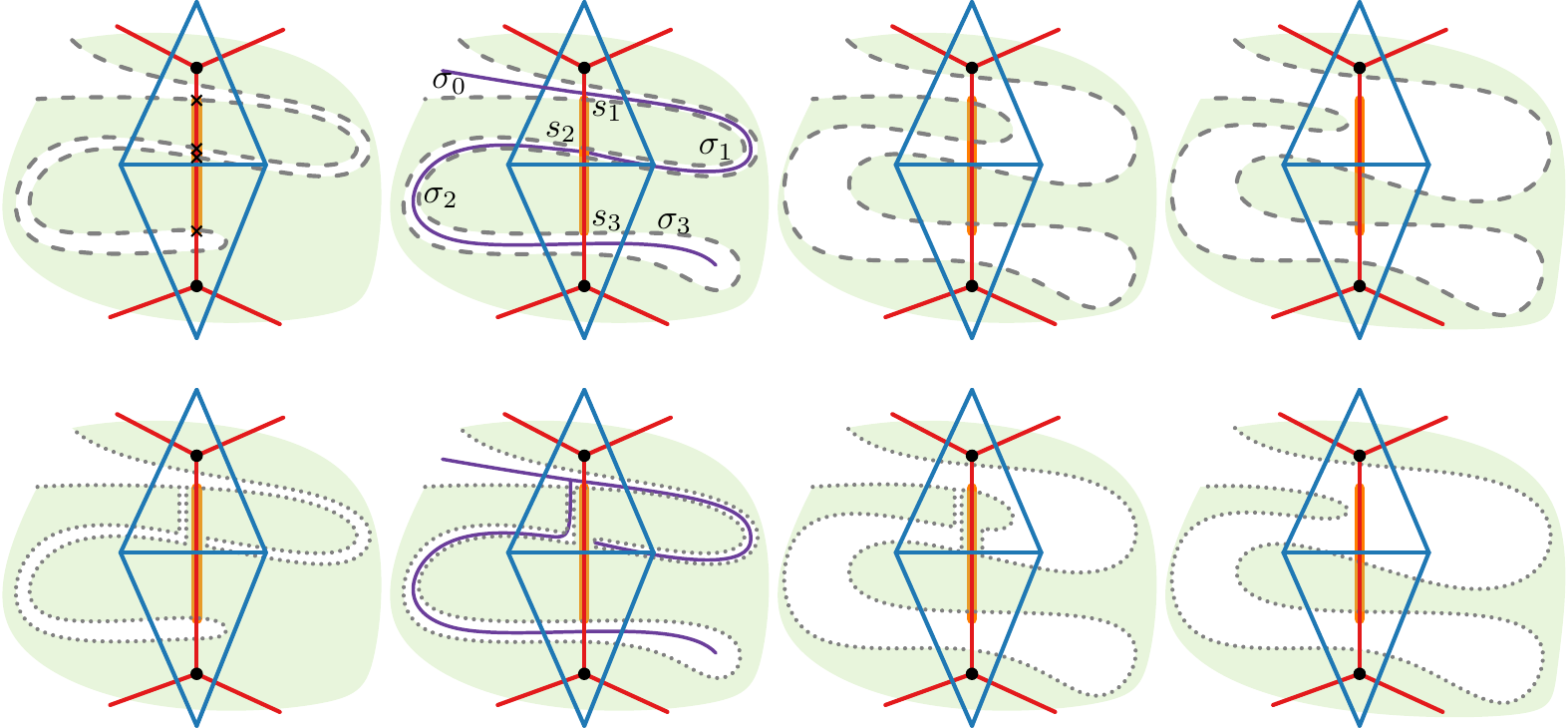}
		\caption{A local surgery to avoid five spikes of the same edge on a single spiked path.}
		\label{F:case2}
	\end{figure}
	\begin{figure}[h]\centering
		\includegraphics[width=\textwidth]{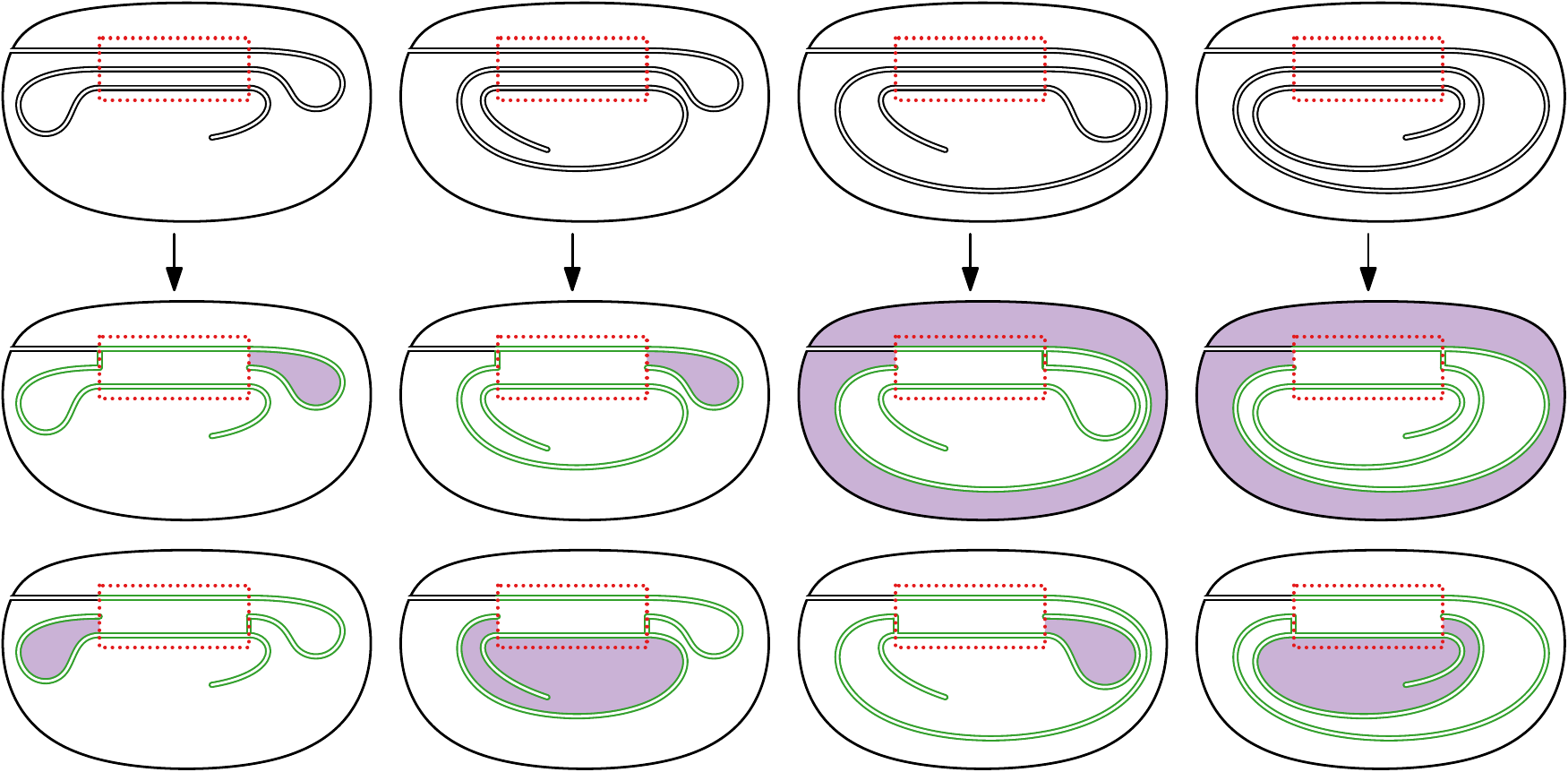}
		\caption{Cases for shortcutting spiked paths visiting the same edge often. The neighborhood of the repeated edge is dotted and the component contracted first is shaded.}
		\label{F:3Edge}
	\end{figure}

		To get rid of spike~$s_2$, we connect~$\sigma_0 s_1 \sigma_1$ to~$\sigma_2 s_3 \sigma_3$ in an alternative way.
		Figure~\ref{F:3Edge} illustrates all possible ways~$s_1$,~$s_2$ and~$s_3$ (in the dotted area) can be connected by~$\sigma$, and how our method will reconnect~$\sigma$ without~$s_2$.
		Formally, to decide where this reconnection takes place, we consider the components of~$A(h(\tau_1),\gamma_1)\setminus\pi$, where~$\pi$ is the arc of~$e^*$ between its intersections with~$s_1$ and~$s_3$.
		There are three components, component~$C_1$ touching~$\pi$ and~$\sigma_1$, component~$C_2$ touching~$\pi$ and~$\sigma_2$, and component~$C_3$ touching~$\sigma$ entirely, and touching~$\pi$ in two arcs.
		The component that~$h$ contracts first is either~$C_1$ or~$C_2$ (since~$C_3$ lies between the other two).

		First consider the case where~$C_1$ is contracted first, then the path~$\sigma_2 s_3 \sigma_3$ starts in the dual face of the endpoint of~$s_2$.
		Note that there is a (zero-length) path between the start or endpoint of~$s_1$ and the endpoint of~$s_2$ because~$s_1$ and~$s_2$ are adjacent along~$e^*$.
		Use this zero-length path to connect~$\sigma_2 s_3 \sigma_3$ to~$\sigma_0 s_1\sigma_1$ at the start or endpoint of~$s_1$ and call the resulting tree~$\lambda$.

		We claim we obtain an optimal monotone isotopy~$h'$ from~$h$ by replacing the spiked path~$\sigma$ by the spiked tree~$\lambda$, and removing the unspike move of~$e^*$ following the contraction of~$C_1$.
		Up until the creation of~$\lambda$, the move sequence is the same as in~$h$.
		Since~$\lambda$ contains a subset (all spikes except~$s_2$) of the spikes of~$\sigma$, the spiked tree can be created without surpassing the height of~$h$.
		After the creation of~$\sigma$ in~$h$ and~$\lambda$ in~$h'$, locally, the level curves of~$h$ and~$h'$ differ only in a small neighborhood of~$\pi$, so that all moves of~$h$ except those crossing~$\pi$ can still be performed in~$h'$.
		Because~$s_2$ is the only spike along~$e^*$ that lies between~$s_1$ and~$s_3$, the next move that crosses~$\pi$ is the unspike move, call it~$z$, following the contraction of~$C_1$.
		The level curve of~$h'$ just before~$z$ is identical to the level curve of~$h$ just after~$z$, so it is safe to omit move~$z$ in~$h'$.
		All subsequent level curves of~$h$ and~$h'$ are identical, so we conclude that~$h'$ is an optimal monotone isotopy (with fewer moves).
		Therefore, the reduced monotone isotopy of~$h'$ has fewer moves, contradicting that~$h$ was good.

		The proof for the case where~$C_2$ contracts first, is symmetrical, except that the spiked tree~$\lambda$ is created differently.
		In this case, we define~$\lambda$ to be~$\sigma_0 s_1 \sigma_1$, whose endpoint is connected to~$\sigma_2 s_3 \sigma_3$ at the start or endpoint of~$s_3$.
		When spiking this tree, the direction of the spikes on~$\sigma_2$ (and sometimes~$\sigma_3$) is reversed, but this does not affect the proof.

		Hence, in a good homotopy, no spiked path spikes the same edge~$5$ times.

	\end{proof}

	\begin{theorem}\label{T:NP}
	For~$\gamma_0$ and~$\gamma_1$ bounding an annulus with $n$ faces and $m$ edges, there is a homotopy of minimum height that has at most $O(mn)$ moves. Therefore, deciding whether their \textsc{Homotopy Height} is at most~$L$ is in \NP.
	\end{theorem}
	\begin{proof}
		Let~$n$ be the number of faces, and~$m$ the number of edges in~$G$.
		As a direct consequence of Lemmas~\ref{L:polynomialfaceflips} and~\ref{L:shortSpikedPaths}, there is a good homotopy that spikes each edge at most~$4(n+1)$ times and unspikes each edge at most~$4(n+1)$ times.
		So there is a homotopy of minimum height that has at most~$8m(n+1)+n=O(mn)$ moves.
		Testing whether this homotopy indeed has height at most~$L$ can be done by computing the maximum length over its (polynomially many) level curves, each containing a polynomial number of edges, and comparing this maximum with~$L$.
		Given a good homotopy, all of this can be done in polynomial time assuming addition and comparisons of numbers takes polynomial time.
	\end{proof}

        We note that the \textsc{Homotopy Height} problem can also be defined in slightly different settings, for example 
        \begin{itemize}
        \item $\gamma_0$ and $\gamma_1$ are two paths with common endpoints $s$ and $t$, such that $\gamma_0 \cup \gamma_1$ is the boundary of a combinatorial disk. Then $\gamma_0$ is homotopic to $\gamma_1$ with fixed endpoints, and we are interested in computing the optimal height of this homotopy. This is the \textsc{Homotopy Height} problem considered by E. Chambers and Letscher~\cite{homotopyheight}.
        \item There is a single curve $\gamma$ forming the boundary of a combinatorial disk. This curve is contractible, and we are interested in computing the optimal height of such a contraction. This is one of the settings considered in~\cite{ccmor-mcbd-17}.
        \end{itemize}

In both these cases, the Theorems~\ref{T:isotopydiscrete} and~\ref{T:monotonediscrete} have analogues establishing that some optimal homotopy is an isotopy and is monotone. The rest of our proof techniques then readily apply, and prove that the \textsc{Homotopy Height} problem in these cases is also in \NP. The next section investigates more distant variants.

\section{Variants and approximation algorithms}\label{S:variants}
\subsection{Homotopic Fr\'echet distance}
There is a strong connection between the problem of \textsc{Homotopy Height} and the problem of \textsc{Homotopic Fr\'echet distance}, which we now recall.  As in~\cite{hnss-hwdmml-16}, our setting is the one of a disk $D$ with four points~$p_0$, $q_0$, $q_1$ and $p_1$ on the boundary, connected by four disjoint boundary arcs~$\gamma_0$,~$\gamma_1$,~$P$ and~$Q$, with~$\gamma_0$ from~$p_0$ to~$q_0$;~$\gamma_1$ from~$p_1$ to~$q_1$;~$P$ from~$p_0$ to~$p_1$; and~$Q$ from~$q_0$ to~$q_1$, see Figure~\ref{F:variants}, left. A homotopy between $\gamma_0$ and $\gamma_1$ is a series of elementary moves connecting curves of $D$ with one endpoint on $P$ and the other on $Q$, where the collection of curves starts at $\gamma_0$ and ends at $\gamma_1$. The \textsc{Homotopic Fr\'echet distance} between $P$ and $Q$ is the height of a homotopy between $\gamma_0$ and $\gamma_1$ of minimal height. The common intuition for this distance is that it is the minimal length of a leash needed for a man on $P$ to walk his dog along $Q$, where the leash may stretch but cannot be lifted out of the underlying space.

We note that this is slightly different than the original setting for homotopic Fr\'echet distance in the original work~\cite{CVE08}, where an exact algorithm is presented for the plane minus a set of polygonal obstacles.
In the original work, the start and end leashes are not fixed, and in fact the bulk of the work is in determining an optimal relative homotopy class in order to find the best homotopy.

	\begin{figure}[ht]\centering%
		\includegraphics{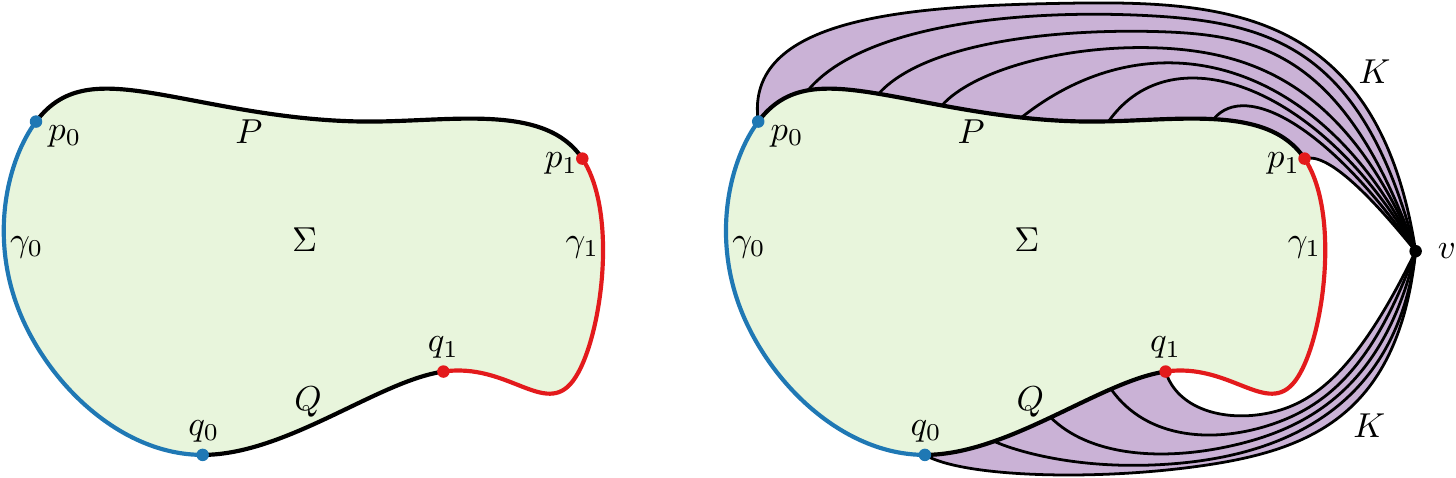}
		\caption{The setting of homotopic Fr\'echet distance.}\label{F:variants}
	\end{figure}%

\begin{proposition}
The \textsc{Homotopic Fr\'echet distance} problem is in \textbf{NP}.
\end{proposition}

\begin{proof}
	We reduce \textsc{Homotopic Fr\'echet Distance} to \textsc{Homotopy
		Height} using the following construction. We add a vertex $v$ and
	edges of weight $K$ between this vertex and all the vertices of the
	paths $P$ and $Q$, where $K$ is a constant greater than the sum of
	the weights of the edges of the disk, as well as all the
	intermediate triangles, see Figure~\ref{F:variants}, right.
	This results in a pinched annulus $A$, with
	two boundaries $\gamma_0'$ and $\gamma_1'$ obtained from the paths
	$\gamma_0$ and $\gamma_1$, both completed into closed curves using the
	additional vertex $v$. We claim that an optimal homotopy between
	$\gamma_0$ and $\gamma_1$ translates into an optimal homotopy in $A$
	between $\gamma_0'$ and $\gamma_1'$, and vice-versa. Indeed, by
	Lemma~\ref{L:spanning}, there exists an optimal homotopy in $A$ such
	that any intermediate curve crosses the shortest path between
	$\gamma'_0$ and $\gamma'_1$ exactly once, and in our case the shortest
	path is the zero length path starting and ending at the vertex
	$v$. Furthermore, if the weight $K$ is big enough, the level curves
	of an optimal homotopy between $\gamma'_0$ and $\gamma'_1$ will
	always use exactly two of the edges of weight $K$, since two are
	needed but any more would be too expensive. Thus, an optimal
	homotopy between $\gamma_0'$ and $\gamma_1'$ translates directly
	into an optimal homotopy between $\gamma_0$ and $\gamma_1$ after
	cutting on $v$ and removing the edges linked to $v$ and
	vice-versa. The homotopy height is increased by exactly $2K$ in
	this translation.
\end{proof}

\noindent
Har-Peled, Nayyeri, Salvatipour and Sidiropoulos~\cite{hnss-hwdmml-16} provide an algorithm to compute in $O(n \log n)$ time a homotopy of height~$O(d\log n)$, where~$d$ is a lower bound on the height of an optimal homotopy, and~$n$ is the complexity of~$\Sigma$.
In particular, one can set~$d$ to be the maximum of~$\|\gamma_0\|$,~$\|\gamma_1\|$, the diameter of~$\Sigma$, and half of the total weight of the boundary of any face.
This yields an~$O(\log n)$ approximation for \textsc{Homotopic Fr\'echet distance}\footnote{This algorithm assumes triangular faces, but using our definition of~$d$, one can extend the algorithm of~\cite{hnss-hwdmml-16} to also work with polygonal faces.}.
We show here that their algorithm can be adapted to yield an $O(\log n)$ approximation for \textsc{Homotopy Height}.

\begin{proposition}\label{C:logn}
  One can compute in $O(n \log n)$ time an $O(\log n)$-approximation of \textsc{Homotopy Height}.
\end{proposition}

\begin{proof}
Starting with an annulus and two boundary curves $\gamma_0$ and
$\gamma_1$, we first compute a shortest path $\mathcal{P}$ between the
boundary curves $\gamma_0$ and $\gamma_1$ and cut along $\mathcal{P}$
to obtain a disk $D$. This brings us to the setting of
\textsc{Homotopic Fr\'echet Distance}, and we can apply the
aforementioned algorithm and obtain a homotopy $h$. In order to
recover a homotopy between $\gamma_0$ and $\gamma_1$, we glue back the
disk along $\mathcal{P}$ into an annulus, and the level curves of $h$
are completed into closed curves by using subpaths of $\mathcal{P}$,
this gives us a homotopy $h'$. It remains to show that this is an
$O(\log n)$ approximation of the optimal homotopy. By
Lemma~\ref{L:spanning}, some optimal homotopy between $\gamma_0$
and $\gamma_1$ has level curves cutting $\mathcal{P}$ exactly
once. Thus, the height $L$ of an optimal homotopy in
the disk $D$ is a lower bound for the height of an optimal homotopy in the annulus $A$. Furthermore, each level curve
$\gamma_t$ of $h'$ consist of two subpaths, one being a level curve
$h(t)$ of $h$ and the other being a subpath $\mathcal{P}'_t$ of
$\mathcal{P}$. Since $\mathcal{P}$ is a shortest path,
$\mathcal{P}'_t$ is also a shortest path between its endpoints, so it
is shorter than $h(t)$ since they have the same endpoints. By
construction, the length of $h(t)$ is $O(L \log n)$, and thus
the length of $\gamma_t$ is $O(2L \log n)=O(L \log n)$. This concludes
the proof.
\end{proof}
\subsection{Minimum height linear layouts}

We also show that a seemingly unrelated graph drawing problem is
directly equivalent to the \textsc{Homotopy Height} problem. A
\emph{linear layout} is an embedding of a planar graph where the edges have isolated tangencies with the vertical line, and all the vertices have
distinct $x$ coordinates. The \textsc{Minimum Height Linear
	Layout} problem is the following one: Given a planar embedding of an edge-weighted
graph $G$, find a homeomorphic linear layout of $G$ in $\R^2$ such that the maximal weight of
the vertical lines is minimized. Here, the weight of a vertical line
is the sum of the weights of the edges that it crosses, and (similarly to the cross-metric setting), vertical
lines crossing tangent to the edges or crossing vertices are not counted.

\begin{theorem}
The \textsc{Minimum Height Linear Layout} problem is equivalent to the \textsc{Homotopy Height} problem.
\end{theorem}

	\begin{figure}[h]\centering%
		\includegraphics{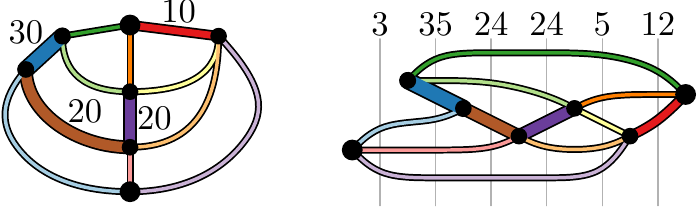}%
		\caption{Dual representation of Figure~\ref{F:G} (left) and Figure~\ref{F:Gh} (right).}%
		\label{F:dualGh}%
	\end{figure}

        \begin{proof}
Indeed, a linear layout of a planar graph $G$ naturally induces a
discrete homotopy sweeping its dual graph $G^*$. More formally, we drill a small hole around the vertex dual to the outer face of $G$, and we view its complement as a disk $D$ which is a cross-metric surface for the graph $G$. Since the hole was drilled in the middle of the face of $G$, its boundary has zero length. We pick two arbitrary vertices $s$ and $t$ on it, which cuts the boundary into two paths $L$ and $R$. Then we claim that a minimum height linear layout of $G$ is equivalent to a homotopy of minimum height between $L$ and $R$ (where the endpoints are fixed)\footnote{The point of the somewhat artificial construction with $L$ and $R$ is to force the homotopy to go through the outer face of $G$ at all times.}. Indeed, whenever
the sweep of $\mathbb{R}^2$ induced by the vertical lines crosses an
edge or passes a vertex, by the dual interpretation of homotopies with
cross-metric surfaces outlined in the preliminaries, it amounts to
doing a face or an edge move, and thus the whole vertical sweep
defines a homotopy between the two paths $L$ and $R$. Furthermore, this homotopy is an isotopy, since the vertical
lines are simple, and a monotone one since they only go in a single
direction. Conversely, a discrete homotopy of optimal height between
$L$ and $R$ can be ``straightened'' into a linear layout: by
Theorem~\ref{T:monotonediscrete}, one can assume such a homotopy $h$ to be an isotopy
and to be monotone, and therefore the succession of dual moves of $h$
with respect to $G$ are homeomorphic to a sweep of $G$ by vertical
lines, as pictured in Figure~\ref{F:dualGh}. An optimal homotopy
amounts, via this homeomorphism, to finding a linear layout of minimal
weight.
\end{proof}

In particular, the \textsc{Minimum Height Linear Layout} problem is in NP and admits an $O(\log n)$-approximation algorithm.

\paragraph*{Acknowledgements.} We are grateful to Tasos Sidiropoulos for his involvement in the early stages of this research, and to Gregory Chambers and Regina Rotman for helpful discussions. 
This research began while partially supported through the program ``Simons Visiting Professorship" by the Mathematisches Forschungsinstitut Oberwolfach in 2015.
Erin Chambers is supported in part by NSF grants IIS-1319944,  CCF-1054779, and CCF-1614562.
Arnaud de Mesmay is partially supported by the ANR project ANR-16-CE40-0009-01 (GATO).
Tim Ophelders is supported by the Netherlands Organisation for Scientific Research (NWO) under project no. 639.023.208.

	\bibliographystyle{siam}
	\bibliography{biblio}
\end{document}